\documentclass[12pt,reqno]{amsart}
\usepackage{amsthm}
\usepackage{amsmath}
\usepackage{latexsym}
\usepackage{amsfonts}
\usepackage{amssymb}
\usepackage{color}
\usepackage{bbm,dsfont}
\usepackage{graphicx}
\usepackage{hyperref}
\usepackage{enumerate}
\usepackage{comment}
\usepackage{cite}
\usepackage{caption}
\usepackage{subcaption}

\newtheorem{proposition}{Proposition}

\theoremstyle{definition}

\newtheorem{example}{Example}
\newtheorem{definition}{Definition}

\newcommand{\blue}[1]{\textcolor{blue}{#1}}

\newcommand{\red}[1]{\textcolor{red}{#1}}

\newcommand{\ip}[2]{\left\langle\,#1\,|\,#2\,\right\rangle} 
\newcommand{\innerp}[2]{\left\langle\,#1\,,\,#2\,\right\rangle} 
\newcommand{\kb}[2]{|#1\rangle\langle#2|} 
\newcommand{\no}[1]{\left\|#1\right\|} 
\newcommand{\tr}[1]{{\rm tr}\left[#1\right]} 

\newcommand{\id}{\mathbbm{1}} 

\renewcommand{\rho}{\varrho}

\newcommand{\abs}[1]{\lvert #1 \rvert} 

\newcommand{\Mo}{\mathsf{M}}

\newcommand{\psuc}{P_{\mathrm{success}}} 
\newcommand{\pfail}{P_{\mathrm{fail}}} 

\newcommand{\fram}{\left\{ f_i \right\}_{i=1}^n}

\begin{document}\setlength{\arraycolsep}{2pt}

\title[]{Partial ignorance communication tasks in quantum theory}

\author{Oskari Kerppo}
\address{Department of Physics and Astronomy, University of Turku, FI-20014 Turku, Finland}

\begin{abstract}

We introduce a generalization of communication of partial ignorance where both parties of a prepare-and-measure setup receive inputs from a third party before a success metric is maximized over the measurements and preparations. Various methods are used to obtain bounds on the success metrics, including SDPs, ultraweak monotones for communication matrices and frame theory for quantum states. Simplest scenarios in the new generalized prepare-and-measure setting, simply called partial ignorance communication tasks, are analysed exhaustively for bits and qudits. Finally, the new generalized setting allows the introduction of operational equivalences to the preparations and measurements, allowing us to analyse and observe a contextual advantage for quantum theory in one of the communication tasks.

\end{abstract}

\maketitle
 
\tableofcontents

\section{Introduction}

In the usual communication setting one party is interested in sending an encoded message to another party. The receiving party will try to decode the message so that information can be transmitted between the parties. In quantum mechanics these parties are typically called Alice and Bob. The limiting factor of this communication is usually the communication medium, which could be a classical bit or the quantum analogue, a qubit, for instance. If Bob successfully decodes the message sent by Alice, the communication between them is deemed successful.

Perhaps the most basic scenario is the following: Alice has a preparation device that can prepare $n$ distinct states of the communication medium. Hence each of the states encodes a unique message. Bob, on the other hand, has a measurement device with $n$ distinct outcomes. If each outcome of Bob's device identifies with certainty the state prepared by Alice, the states are called \textit{distinguishable}. The Basic Decoding Theorem \cite{ScWe2010} states that, in the quantum case, whenever there are more possible messages than distinguishable states the error in this kind of communication is at least $1-\frac{d}{n}$, where $d$ is the respective Hilbert space dimension of the quantum state, and $n$ is the number of states in Alice's preparation device.

In a recent work \cite{HeKe2019}, a variation of the basic communication scenario was studied. In this variant there is a third party, Charlie, who acts as a game master to the following communication game. Charlie has $n$ empty boxes, and he hides a prize inside one of them. He reveals at least one empty box to Alice, who must then communicate this information to Bob. Alice and Bob win the game if Bob chooses the box with the prize in it. The communication between Alice and Bob is successful if Bob has the same chance of finding the prize as Alice would. Hence Alice's encoded message should contain information on which choices Bob should avoid, as even Alice doesn't know where the prize is. This communication of choices to be avoided was called \textit{communication of partial ignorance} by the authors. The qubit case was perfectly characterized, and some general results concerning qudits were presented in \cite{HeKe2019}. The main takeaway was that, perhaps remarkably, perfect communication of partial ignorance requires an entirely different setup than the basic communication scenario, where Bob tries to identify the state sent by Alice directly.

In the present article, we continue the work done in \cite{HeKe2019} by modifying the communication game in a major way. In this new game it's not only Alice who is revealed information about the location of the prize, but also Bob. Charlie could, for instance, reveal one empty box to Alice and one empty box to Bob. This seemingly simple modification has dramatic effects on the strategies that Alice and Bob must employ. Moreover, the difficulty of finding viable communication strategies and proving their optimality is significantly increased. Typically Bob will have multiple choices for his measurement based on the information that is revealed to him. Thus it will become possible to analyse the effects of contextuality as a resource to this game. In general, entirely different mathematical tools are required to analyse this new scenario compared to the communication of partial ignorance.

The rest of the article is structured as follows. In Sec. \ref{comm-tasks} the new communication game is presented in more detail. We then go on to present the required mathematical methods in Sec. \ref{methods}. A brief overview to semidefinite programming, contextuality and frame theory is given. Using the introduced mathematical methods we will analyse the communication game in Sec. \ref{bounds}. Numerical and analytical bounds on successful communication will be presented in the classical, contextual and general quantum case. Finally in Sec. \ref{final-thoughts} we end the present article with some final thoughts and discussion on possible future directions.

\section{Partial ignorance communication tasks}\label{comm-tasks}

Let's consider a general communication setting first. Suppose Charlie has an $n$-bit string, $s$. The objective of the communication task is that Bob must guess the value of at least one bit in $s$. Alice will receive a bit string, or input, of her own, $a$. The bit string $a$ can be of any fixed but finite length, and in general it will contain some information on the bit string $s$. Alice then sends an encoded message, $r(a)$, to Bob. In the classical case the message will be a bit string of fixed length. A qudit $\rho_a$ of fixed dimension $d$ will be sent in the quantum case. Upon receiving the message $r(a)$ sent by Alice, Bob will also receive a bit string $b$ of fixed length from Charlie. Once communication is done Bob must produce a guess on the values of at least some bits in $s$.

An example of the above communication setting would be a random access code (RAC) \cite{Wiesner1983, AmNaTaVa2002, AmLeMaOz2009}. In a RAC Alice's input $a$ will coincide with Charlie's bit string $s$. Typically Alice is only allowed to send one bit or a qubit to Bob. The input $b$ given to Bob will contain the index of the bit whose value Bob must guess. If we denote by $g(b, r(a))$ the guess Bob produces based on his input $b$ and the message $r(a)$, Alice and Bob win the game if $g(b, r(a)) = s_b$, where $s_b$ denotes the value of the bit string $s$ at index $b$. 

Communication of partial ignorance, with the notation we have been using, would be presented in the following way. Charlie has an $n$-bit string $s$ containing exactly one 1, the index of which will indicate the location of the prize. Alice's input $a$ will reveal the index of at least one of the 0's. Bob doesn't receive any input except for the message sent by Alice. Alice and Bob win if Bob's guess $g(r(a))$ equals the index of the 1 in $s$.

\begin{figure}
    \centering
    \includegraphics[width=\textwidth]{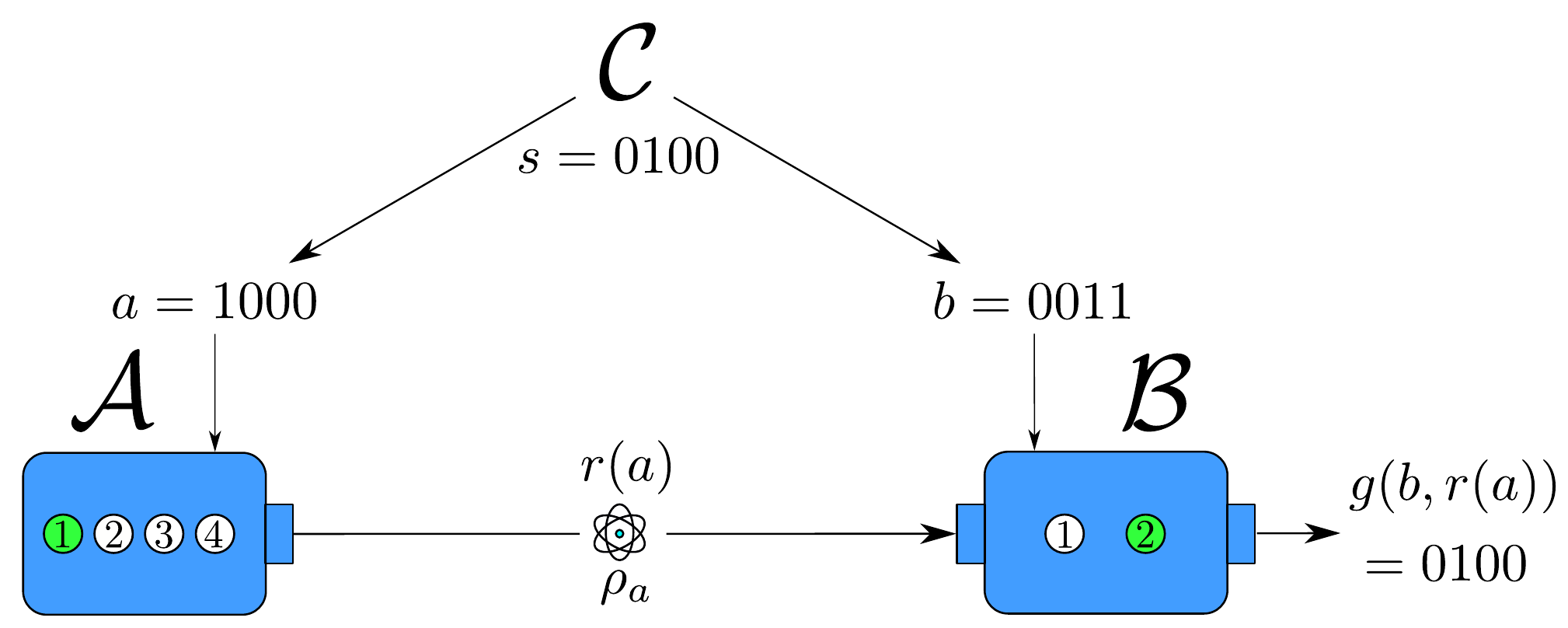}
    \caption{An exemplary setting of a partial ignorance communication task of type $T_{4,1}$. The 1 in $a$ reveals the first 0 in $s$. The two 1's in $b$ reveal the last two 0's in $s$. Bob should guess the correct index of the 1 in $s$.}
    \label{fig:commsetting}
\end{figure}

It's very important to note that, in communication of partial ignorance, the game master Charlie is not obliged to sample the bit-string $s$ according to any fixed distribution. In fact, he is allowed to freely choose it even after Alice has already received her input. However, he must not give conflicting information to Alice. Thus perfect communication of partial ignorance requires Alice and Bob to minimize the worst case error probability with respect to all possible inputs. The strategy that Alice and Bob must implement leads to the concept of communication matrices, which were extensively studied in \cite{HeKeLe2020} along with their operational hierarchy.

We can now define the generalization of communication of partial ignorance we are interested in.

\begin{definition}\label{def:singleshotcommtask}
A partial ignorance communication task of type $T_{n,m}$ is a communication game where a game master Charlie chooses an $n$-bit string $s$ with exactly one 1. Charlie then sends an input $a$ revealing the indices of $m$ 0's in $s$ to Alice. Bob will be revealed the remaining indices of 0's in input $b$. Both inputs $a$ and $b$ are bit strings of same length as $s$, and a value 1 in both inputs will reveal a 0 in $s$ at the corresponding index. Alice and Bob are allowed to communicate according to predefined rules. After communication Bob must produce a guess on the index of the 1 in $s$. Alice and Bob win if Bob's guess is correct.
\end{definition}

An illustration of the partial ignorance communication task of type $T_{4,1}$ is given in Fig. \ref{fig:commsetting}.

The limiting factor in partial ignorance communication tasks is the communication medium. In the classical version Alice is allowed to send a classical message $r(a)$ containing a predefined number of bits to Bob. In the cases studied in this article the number of bits will be just one. In the quantum version a qudit $\rho_a$ of fixed dimension can be sent.

It should be noted that any type of partial ignorance communication task will be a game of complete information in the sense that Alice and Bob's inputs together will reveal the index of the 1. We could also study the case where this is not the case, but for now we will only study the informationally complete version. For the remainder of this article, whenever we are talking about a task $T_{n,m}$, we will always be referring to the partial ignorance communication task of type $T_{n,m}$ as defined in Def. \ref{def:singleshotcommtask}.

Let's consider the simplest task $T_{3,1}$ as an example before moving on to the mathematical methods.

\subsection{Classical and quantum strategies for \texorpdfstring{$T_{3,1}$}{T3,1}}

Let's begin by introducing the best possible classical strategy for the task $T_{3,1}$.

\begin{example}\textit{(Classical strategy for $T_{3,1}$)}
There are three different choices for the string $s$, namely $100$, $010$ and $001$. These strings also coincide with the possible inputs $a$ and $b$ for Alice and Bob. Tab. \ref{tab:comm-task-c31-strategy} presents a general form of the strategy that Alice and Bob will implement. A specific strategy is obtained by replacing the variables $i$, $j$ and $k$ in Tab. \ref{tab:comm-task-c31-strategy} by values 0 or 1, and by listing the values of the guesses $g(b, r(a))$ in the middle table.

By looking at the rows of Tab. \ref{tab:comm-task-c31-strategy} where the input $b$ is identical, it becomes apparent that in order to never make a mistake the variables $i$, $j$ and $k$ should fulfil the following restrictions: \begin{align*}
    j \neq i, \quad k \neq i, \quad k \neq j.
\end{align*} Otherwise Bob is forced to produce the same guess with different strings $s$ because his inputs are identical. However, the variables $i$, $j$ and $k$ can only take on the binary values of 0 or 1, and hence at least one of the restrictions must be violated. This means that any classical strategy, formed by giving concrete values to Tab. \ref{tab:comm-task-c31-strategy}, must contain at least one mistake. The average classical success probability in task $T_{3,1}$ is therefore at most $5/6$. Tab. \ref{tab:c31-optimal-strategy} shows an example of an optimal classical strategy that saturates the upper bound for success probability. Therefore the classical average success probability for $T_{3,1}$ is exactly $5/6$.

\begin{table}
\centering
    \begin{tabular}{c|c}
        $a$ & $r(a)$ \\ \hline
        $100$ & $i$  \\
        $010$ & $j$  \\
        $001$ & $k$
    \end{tabular}
    \begin{tabular}{c|c|c}
        $b$ & $g(b, 0)$ & $g(b,1)$ \\ \hline
        $100$ & $l$ & $x$ \\
        $010$ & $m$ & $y$ \\
        $001$ & $n$ & $z$
    \end{tabular}
    \begin{tabular}{c|c|c|c|c}
        $s$ & $a$ & $b$ & $r(a)$ & $g(b, r(a))$ \\ \hline
        $100$ & $010$ & $001$ & $j$ & $g(001, j)$ \\ 
        $100$ & $001$ & $010$ & $k$ & $g(010, k)$ \\
        $010$ & $100$ & $001$ & $i$ & $g(001, i)$ \\
        $010$ & $001$ & $100$ & $k$ & $g(100, k)$ \\
        $001$ & $100$ & $010$ & $i$ & $g(010, i)$ \\
        $001$ & $010$ & $100$ & $j$ & $g(100, j)$ 
\end{tabular}\newline
\caption{General classical strategy for the task $T_{3,1}$.}
\label{tab:comm-task-c31-strategy}
\end{table}

\begin{table}
\centering
    \begin{tabular}{c|c}
        $a$ & $r(a)$ \\ \hline
        $100$ & $1$  \\
        $010$ & $1$  \\
        $001$ & $0$
    \end{tabular}
    \begin{tabular}{c|c|c}
        $b$ & $g(b,0)$ & $g(b,1)$ \\ \hline
        $100$ & $2$ & $3$ \\
        $010$ & $1$ & $3$ \\
        $001$ & $-$ & $1$
    \end{tabular}
    \begin{tabular}{c|c|c|c|c}
        $s$ & $a$ & $b$ & $r(a)$ & $g(b, r(a)$ \\ \hline
        $100$ & $010$ & $001$ & $1$ & $1$ \\ 
        $100$ & $001$ & $010$ & $0$ & $1$ \\
        $010$ & $100$ & $001$ & $1$ & \red{1} \\
        $010$ & $001$ & $100$ & $0$ & $2$ \\
        $001$ & $100$ & $010$ & $1$ & $3$ \\
        $001$ & $010$ & $100$ & $1$ & $3$ 
\end{tabular}\newline
\caption{An optimal classical strategy for the task $T_{3,1}$. The erroneous value $g(001, 1)=1$ is shown in red. The guess $g(001, 0)$ is not listed in the middle table because it is not needed in this implementation of the strategy.}
\label{tab:c31-optimal-strategy}
\end{table}

\end{example}

Note that the worst case success probability for all classical strategies for the task $T_{3,1}$ is always zero. Interestingly, the strategy listed in  Tab. \ref{tab:c31-optimal-strategy} has unused capacity for communication, because the guess $g(001, 0)$ never occurs. Perfect communication is nevertheless not possible.

Before moving on to the next section, let us introduce the optimal quantum strategy for $T_{3,1}$.

\begin{example}{\textit{(Quantum strategy for $T_{3,1}$)}}
No matter which string $s$ Charlie chooses, Bob's measurement will try to distinguish between a pair of states. For instance, upon receiving input $001$ from Charlie, Bob's measurement should try to distinguish Alice's states $\rho_1$ and $\rho_2$ because Bob knows Alice is going to prepare one of those states. Following this logic, Alice's preparation device should contain states that can be pairwise distinguished as well as possible. With qubits these states can be chosen from a Bloch sphere plane with equal angles between the state vectors.

We can choose Alice's first state to correspond to the Bloch vector $\begin{bmatrix} 0 & 0 & 1 \end{bmatrix}$. The other two states are obtained by 120 degree clockwise rotations in the $xz$-plane. As density matrices these states can be written as:

\begin{align}
    \rho_1 = \begin{pmatrix} 1 & 0 \\ 0 & 0 \end{pmatrix}, \quad \rho_2 = \frac 14 \begin{pmatrix} 1 & \sqrt 3 \\ \sqrt 3 & 3 \end{pmatrix}, \quad \rho_3 = \frac 14 \begin{pmatrix} 1 & -\sqrt 3 \\ - \sqrt 3 & 3 \end{pmatrix}  
\end{align}

Bob's measurements that best distinguish Alice's states in a pairwise manner, written in terms of POVMs, are the following: \begin{align}\begin{split}
    \Mo_1(1)=\frac12 \begin{pmatrix}1&-1\\-1&1 \end{pmatrix},&\quad \Mo_1(2)=\frac12 \begin{pmatrix}1&1\\1&1 \end{pmatrix}\\
    \Mo_2(1)=\frac12 \begin{pmatrix}1-\frac{\sqrt3}{2}&-\frac12\\-\frac12&1+\frac{\sqrt3}{2} \end{pmatrix},&\quad \Mo_2(2)=\frac12 \begin{pmatrix}1+\frac{\sqrt3}{2}&\frac12\\\frac12&1-\frac{\sqrt3}{2} \end{pmatrix}\\
    \Mo_3(1)=\frac12 \begin{pmatrix}1-\frac{\sqrt3}{2}&\frac12\\\frac12&1+\frac{\sqrt3}{2} \end{pmatrix},&\quad \Mo_3(2)=\frac12 \begin{pmatrix}1+\frac{\sqrt3}{2}&-\frac12\\-\frac12&1-\frac{\sqrt3}{2} \end{pmatrix}.    
\end{split}\end{align} Each POVM is obtained by a rotation of $\pm 90$ degrees in the $xz$-plane from the corresponding state vector, e.g., the effects of POVM $M_2$ are obtained by rotating the Bloch vector of $\rho_2$. For the success probabilities we can calculate that, for instance, $\tr{\rho_3 \Mo_1(1)}=\frac12 \left(1 + \frac{\sqrt3}{2}\right)\approx 0.933013$ and $\tr{\rho_3 \Mo_1(2)} =\frac12 \left(1 - \frac{\sqrt3}{2}\right)\approx 0.0669873$. Because the states and POVMs were chosen in a symmetrical manner we can conclude that the average success probability is equal to $\frac12 \left(1 + \frac{\sqrt3}{2}\right)\approx 0.933013$ while the worst case error probability is equal to $\frac12 \left(1 - \frac{\sqrt3}{2}\right)\approx 0.0669873$.

\begin{figure}
     \centering
     \begin{subfigure}[b]{0.45\textwidth}
         \centering
         \includegraphics[width=\textwidth]{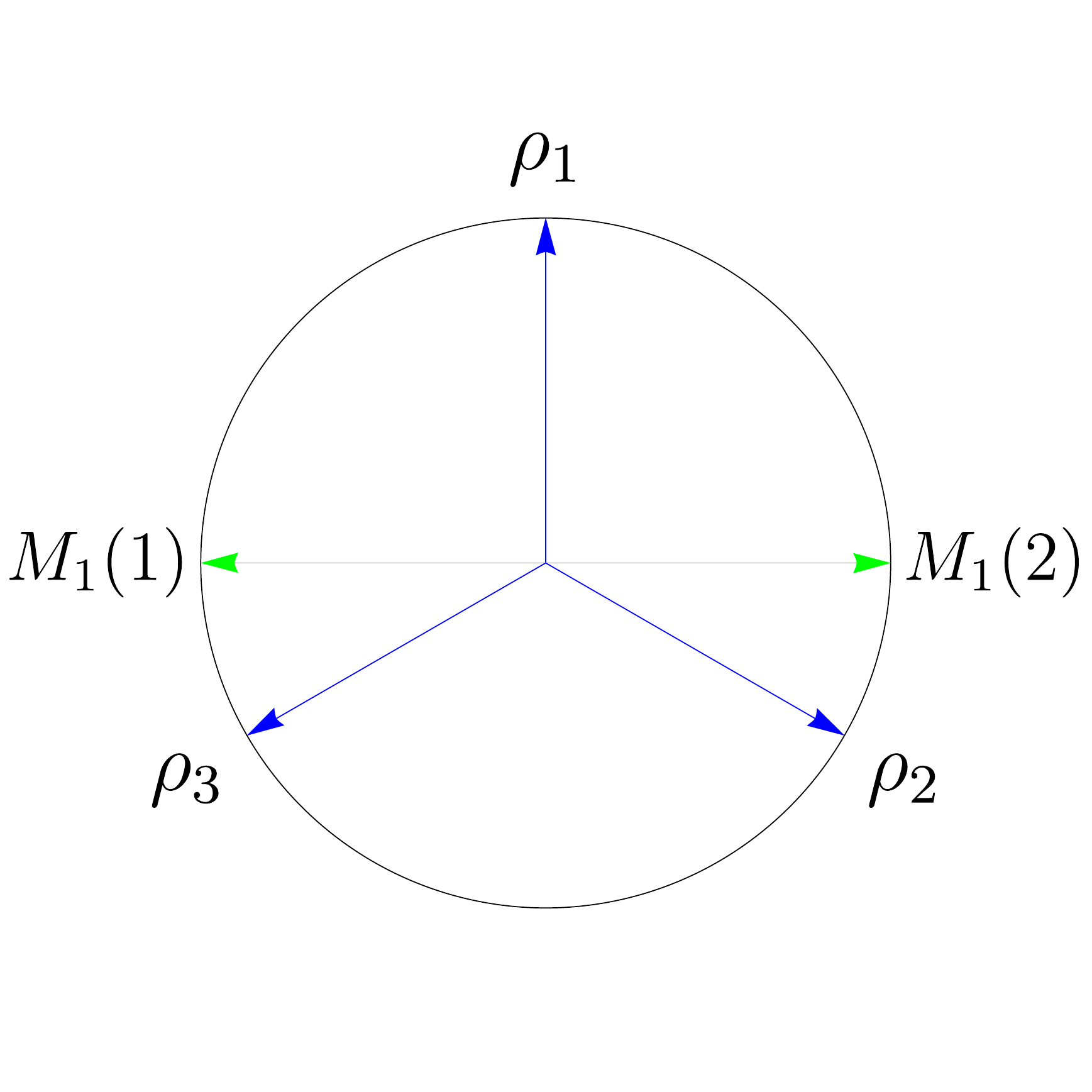}
         \caption{Measurement for $b=100$.}
         \label{fig:mes-c31m1}
     \end{subfigure}
     \hfill
     \begin{subfigure}[b]{0.45\textwidth}
         \centering
         \includegraphics[width=\textwidth]{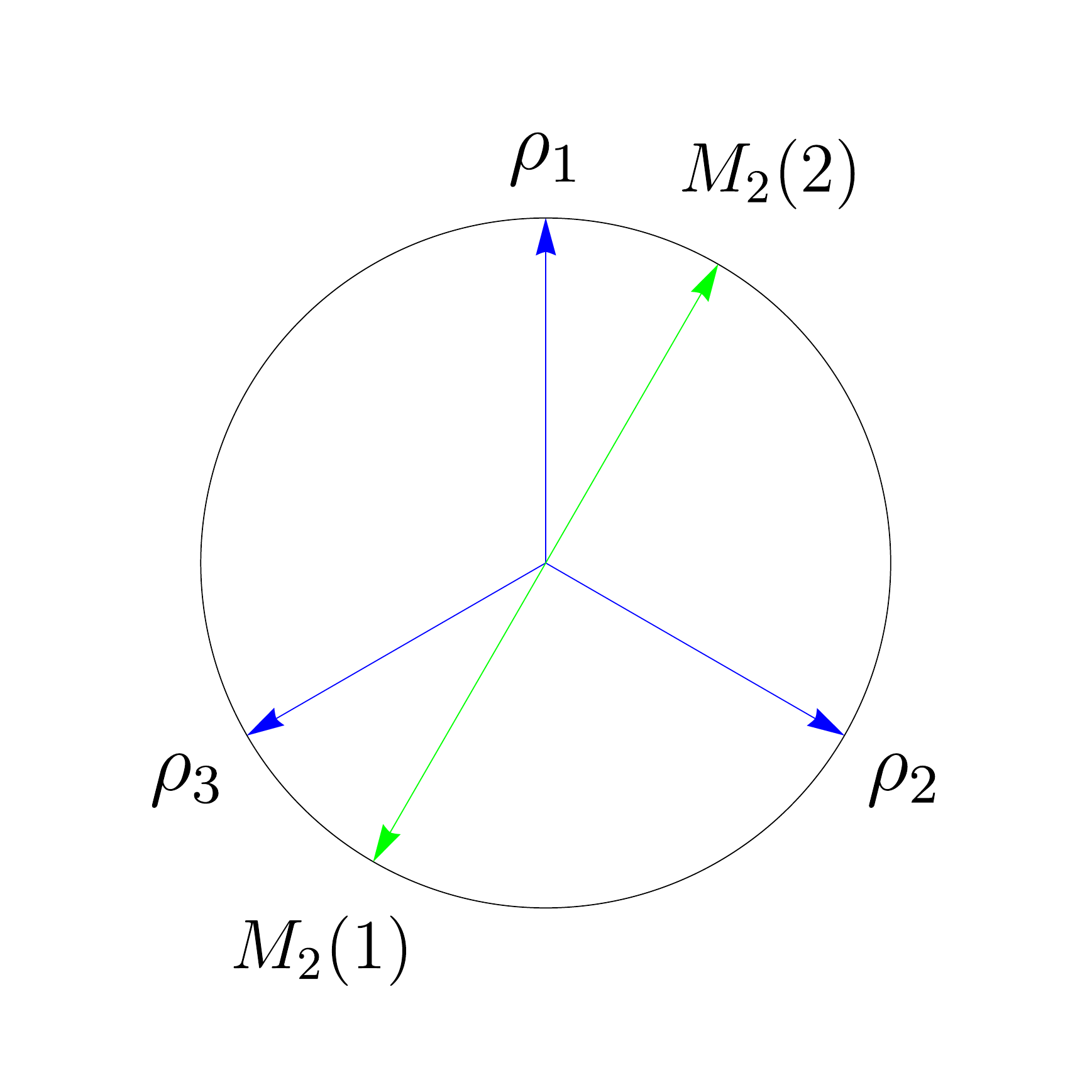}
         \caption{Measurement for $b=010$.}
         \label{fig:mes-c31m2}
     \end{subfigure}
     \hfill
     \begin{subfigure}[b]{0.45\textwidth}
         \centering
         \includegraphics[width=\textwidth]{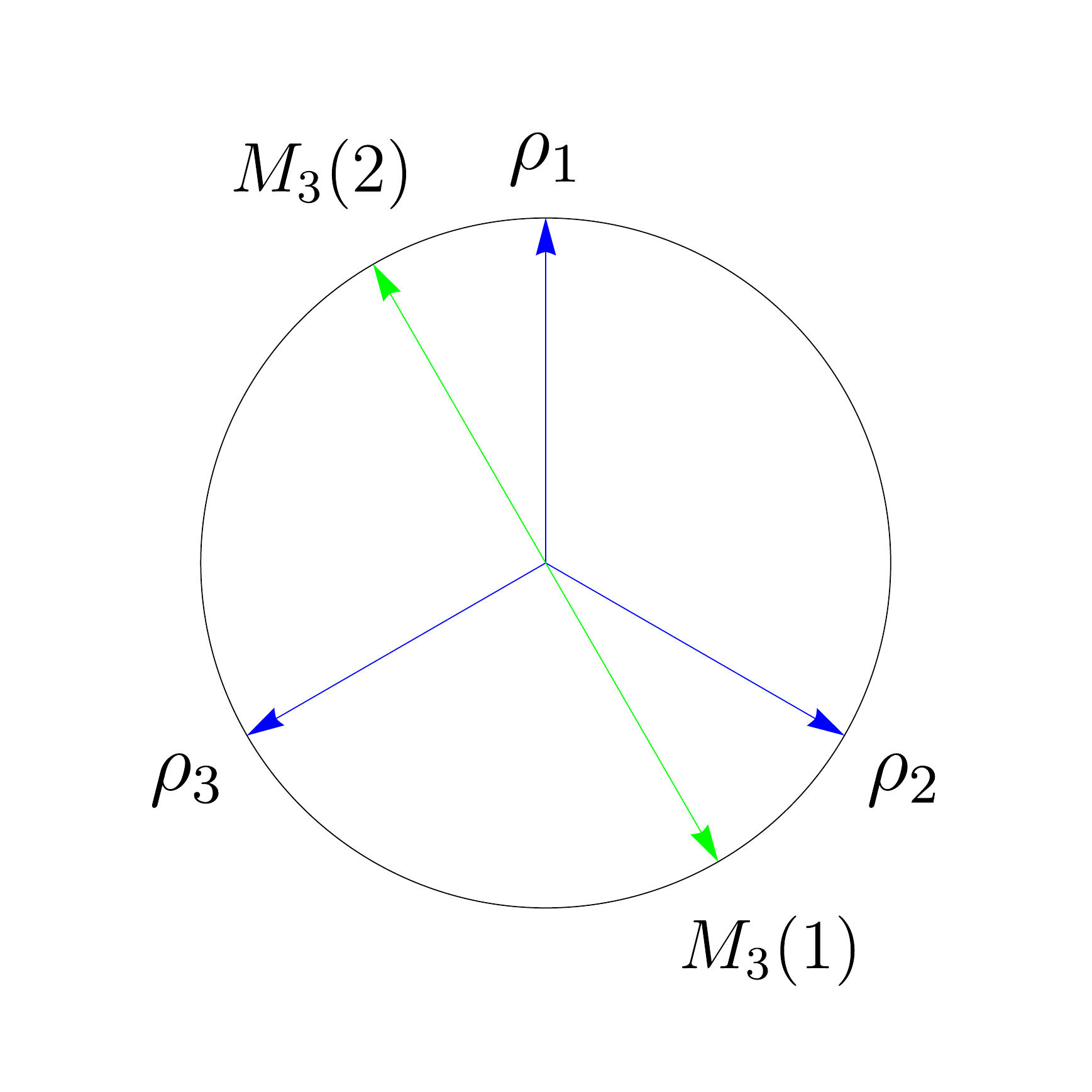}
         \caption{Measurement for $b=001$.}
         \label{fig:mes-c31m3}
     \end{subfigure}
        \caption{Alice's states and Bob's measurements for the task $T_{3,1}$, presented in the $xz$-plane of the Bloch sphere.}
        \label{fig:c31-bobs-measurements}
\end{figure}
\end{example}

As we can see, there is a drastic difference in average and worst case success probabilities between the bit and the qubit in the task $T_{3,1}$. This case was, however, quite easy to analyse. In order to analyse any task $T_{n,m}$ with $n > 3$ we are going to need some mathematical tools.

\section{Mathematical Methods}\label{methods}

In this section we will introduce three tools that can be used to analyse the communication tasks presented in the previous section. Semidefinite programming is a powerful numerical tool that can be used to obtain numerical bounds on success probabilities in various communication and computation tasks involving quantum resources. Contextuality is one such resource and we will use it to provide an alternative comparison between classical and quantum strategies. While semidefinite programs (SDPs) provide valuable insight to many problems, it turns out we can also explain some bounds obtained by SDPs analytically with the help of frame theory. To this end we will use Grassmannian frames.

\subsection{Semidefinite programming}

There are many excellent sources on the theory of SDPs \cite{NaPiAc2007, NaPiAc2008, TaCrUoAb2020, Wi2015, Pi2018, TaZaWoPi2022, NaVe2015, NaFeArVe2015, WaPrLaVaLi2019}. We will be using the unitary SDP hierarchy introduced in \cite{ChFaWr2020} largely due to the fact that it gives very good results already on the first level.

Consider the following optimization problem: \begin{align}\label{eq:optimization}
\begin{split}
\max & \sum_{i,j,k} p_{ijk} \, \tr{\rho_i \Mo_j(k)} \\
s.t. & \quad \rho_i \geq 0 \, \forall i \\
 & \tr{\rho_i} = 1 \, \forall i \\
 & \sum_i \left(  \alpha_i^r - \beta_i^r \right) \rho_i = 0 \, \forall r \in \mathcal{O}_P \\
 & 0 \leq \Mo_j(k) \leq 1 \, \forall j,k \\
 & \sum_k \Mo_j(k) = \id \, \forall j \\
 & \sum_{i,j} \left(  \alpha_{i,j}^s - \beta_{i,j}^s \right)\Mo_i(j) = 0 \, \forall s \in \mathcal{O}_\Mo
\end{split}
\end{align}where the weights $p_{ijk}$ define a success metric on the behavior $p(k|\rho_i , \Mo_j)\equiv \tr{\rho_i \Mo_j(k)}$ and $\mathcal{O}_P$ (indexed by $r$) contains restrictions of the form $\sum_i \alpha_i^r \rho_i = \sum_j \beta_j^r \rho_j$ on  the preparations for some convex weights $\alpha_i^r$ and $\beta_i^r$. Similarly $\mathcal{O}_\Mo$ (indexed by $s$) contains restrictions on the effects for some convex weights $\alpha_{i,j}^s$ and $\beta_{i,j}^s$. All optimal quantum strategies for communication tasks $T_{n,m}$ can be obtained by solving an optimization problem of the form \eqref{eq:optimization}. Notice that the sets $\mathcal{O}_P$ and $\mathcal{O}_\Mo$ may be empty. Unfortunately, this problem is not an SDP. We can, however, use the following algorithm to give the success metric $\sum_{i,j,k} p_{ijk} \, \tr{\rho_i \Mo_j(k)}$ a dimension based lower bound.

\begin{definition}
See-saw algorithm. \begin{itemize}
    \item[1.] Fix a dimension $d$ for the Hilbert space of the quantum states and POVMs.
    \item[2.] Generate random rank-1 states and fix the states $\rho_i$ as constants.
    \item[3.] Optimize Eq. \eqref{eq:optimization} as an SDP over POVM elements $\Mo_j(k)$. Calculate $x_1=\sum_{i,j,k} p_{ijk} \, \tr{\rho_i \Mo_j(k)}$.
    \item[4.] Fix POVM elements $\Mo_j(k)$ as constants with values obtained in the previous step. Optimize Eq. \eqref{eq:optimization} as an SDP over the states $\rho_i$. Calculate $x_2=\sum_{i,j,k} p_{ijk} \, \tr{\rho_i \Mo_j(k)}$.
    \item[5.] Repeat steps 3 and 4 until a predefined threshold value $x_2-x_1 < \epsilon$ is achieved. The value $x_2$ is now a dimension based lower bound on the success metric of Eq. \eqref{eq:optimization}.
\end{itemize}
\end{definition}

To the best of our knowledge the see-saw algorithm lacks any kind of serious theory behind it. Therefore there are no performance guarantees. The see-saw algorithm offers a dimension based lower bound, or the inner bound from here on, and an explicit construction for the states and POVMs that achieve this bound. Some mentions of this method can be found in the literature \cite{WeWo2001, AmBaChKrRa2019, LiDo2006, LiDo2007, LiLiDe2009}.

To obtain a result of optimality we need to produce an upper bound on the success metric of Eq. \eqref{eq:optimization}, or an outer bound. This can be done in the following way.

Suppose we have a quantum behavior $p(k|\rho_i , \Mo_j) = \tr{\rho_i \Mo_j(k)}$. As was proven in \cite{ChFaWr2020}, we can always find unitary matrices such that \begin{align}
    \Mo_j(k) = \frac 12 \id + \frac{U^j_k + U^{j\dag}_k}{4}.
\end{align}
The moment matrices $(\Gamma_i^\mathcal{O})_{j,k} = (\Gamma_i^\mathcal{O})_{\mathcal{O}_j,\mathcal{O}_k} \equiv \tr{\rho_i\mathcal{O}_j^\dag \mathcal{O}_k}$ (with monomials of $\id$ and the unitary operators $U^j_k$ and $U^{j\dag}_k$) have the following properties:\begin{align}\label{eq:unitary-hierarchy}\begin{split}
    \forall i : & \, \Gamma_i^\mathcal{O} \geq 0 \\
    \forall i, r \in \mathcal{O}_P: & \,  \sum_j (\alpha_j^r - \beta_j^r) \Gamma_i^\mathcal{O} = 0 \\
    \forall i,j,k: & \, \left( \Gamma_i^\mathcal{O} \right)_{\id, U_k^j} + \left( \Gamma_i^\mathcal{O} \right)_{\id, U_k^{j\dag}} = 4\left( p(k|i, j) - \frac 12 \right) \\
    \forall i, j : & \, \left( \Gamma_i^\mathcal{O} \right)_{j,j} = 1.
    \end{split}
\end{align}
Additionally, possible constraints from $\mathcal{O}_\Mo$ must be encoded into the moment matrices via conditions on the unitary matrices.

The existence of such moment matrices is a necessary condition for the behavior $p(k|\rho_i , \Mo_j) = \tr{\rho_i \Mo_j(k)}$ to be realizable in a quantum experiment. Moreover, the existence of such moment matrices is a semidefinite feasibility problem, which can be solved with efficient methods \cite{BoVa2004}. The optimization of any success metric over feasible behaviors is again an SDP.

The semidefinite feasibility problem defined by Eq. \eqref{eq:unitary-hierarchy} defines a first level in a hierarchy of SDPs. That is, a behavior $p(k|i,j)$ obtained from maximising a success metric over the feasible set is not necessarily realizable in any quantum experiment. What we do get is an upper bound on the success metric -- no quantum behavior can exceed this bound. Moreover, by considering monomials of length greater than one we obtain a converging hierarchy that converges on the set of quantum behaviors \cite{ChFaWr2020}. In practise we get very good results already on the first level. The notation $\mathcal{U}_1$ is used to specify that a solution is obtained by using the first level of the unitary hierarchy.

\subsection{Noncontextual polytope of correlations}

Contextuality is understood today as an important resource for quantum computation \cite{Bell1966, KoSp1967, Sp2005, SpBuKeToPr2009, MaPuKuReSp2016, ScSpWo2018, ScSp2018, SaCh2019, KuLoPu2019, ChFaWr2020, TaUo2020}. We will now give a brief introduction to the topic, focusing on recent developments. The reader is encouraged to check \cite{PuBaRu2012, Leifer2014} for a more complete introduction to the framework of ontological models.

Suppose there is a state space $\Lambda$, called the ontic state space. Every time a quantum state is prepared, a state $\lambda \in \Lambda$ is produced according to some probability distribution. The ontic states are considered complete descriptions of nature but they are generally speaking inaccessible to all observers. Therefore we associate each preparation $P$ with a corresponding probability measure $\mu_P$ over the state space $\Lambda$. 

When a measurement of some POVM $\Mo$ is performed, it is supposed that the value $\lambda$ completely determines the outcome. Note that this doesn't imply determinism. Instead we associate each effect $\Mo(i)$ with response function: \begin{align}\label{eq:born-rule}
    \int_\Lambda \xi_{\Mo}(i | \lambda ) \mu_\rho(\lambda)d\lambda = \tr{\rho \Mo(i)}.
\end{align}
The response function $\xi_{\Mo}(i | \lambda )$ determines the probability of obtaining outcome $i$ when $\rho$ was prepared and a measurement of $\Mo$ was preformed.

\begin{definition}
An ontological model consist of a measure space $(\Lambda , d\lambda)$\footnote{It would be more rigorous to say that the measure space is $(\Lambda , \Sigma)$ where $\Sigma$ is a $\sigma$-algebra. Here we assume the existence of a canonical measure $d\lambda$ that dominates each probability measure defined on $(\Lambda, \Sigma)$. There are some ontological models that don't allow this, but this assumption makes the notation somewhat simpler.} and two functions $\Delta$ and $\Xi$. The function $\Delta$ maps every quantum state $\rho$ to a set of probability measures $\Delta_\rho$. Likewise the function $\Xi$ maps every POVM $\Mo$ to a set of indicator functions $\Xi_\Mo$. Moreover, for each $\lambda \in \Lambda$ and for every $\Mo$ and $\xi \in \Xi_\Mo$: \begin{align}
    \sum_i \xi_{\Mo}(i | \lambda) = 1.
\end{align}
The ontological model is said to reproduce quantum predictions if it respects the Born rule defined in Eq. \eqref{eq:born-rule}. 
\end{definition}

Quantum theory has the property that mixed states don't have unique decompositions into pure states. Suppose we have a preparation device that can prepare a state $\rho = \frac 12 (\rho_1 + \rho_2) = \frac 12 (\rho_3 + \rho_4)$ through two distinct decompositions. It is a basic fact of quantum information theory that these decompositions are indistinguishable from each other. The principle of noncontextuality states that whenever two states are operationally indistinguishable, they should also be ontologically indistinguishable. On the ontological level this means that the probability measure associated with $\rho$ should be unique. For measurements the principle of noncontextuality states that each effect should be associated with a unique response function, no matter which POVM the effect is a part of. 

It is widely believed, in light of quantum theory, that nature doesn't allow a deterministic description. Therefore the non-classical features of quantum mechanics, such as contextuality and nonlocality, only manifest themselves in a statistical manner. The question is then how can contextuality be detected from a set of experimental data. Let $p(x|P , \Mo)$ denote the probability distribution of outcomes $x$ given that a measurement of $\Mo$ followed the preparation procedure $P$. 

\begin{definition}
Two preparation procedures $P_1$ and $P_2$ are operationally equivalent if $p(x|P_1 , \Mo) = p(x|P_2, \Mo)$ for all $\Mo$. Likewise, two measurements $\Mo_1$ and $\Mo_2$ are operationally equivalent if $p(x| P, \Mo_1) = p(x|P , \Mo_2)$ for all $P$.
\end{definition}

The notation $"\simeq "$ is used to denote operationally equivalent procedures for both preparations and measurements.

Whenever two procedures are operationally equivalent, the principle of noncontextuality implies that the procedures should also be ontologically equivalent. For preparations this means that $\mu_{P_1}(\lambda ) = \mu_{P_2}(\lambda )$ for all $\lambda$ whenever $P_1 \simeq P_2$. For measurements the principle on noncontextuality implies that operationally equivalent effects should be represented by a unique response function: \begin{align*}
    \xi_{\Mo_1}(i|\lambda ) = \xi_{\Mo_2}(i|\lambda ) 
\end{align*} for all $\lambda$ whenever $\Mo_1(i) \simeq \Mo_2(i)$.

It is convenient to collect all operational equivalences in two distinct sets: $\mathcal{O}_P$ will contain all operational equivalences of the form \begin{align*}
    \sum_i \alpha^r_i P_i = \sum_j \beta^r_j P_j
\end{align*} for some distinct sets of convex weights $\{\alpha^r_i\}_i$ and $\{\beta^r_j\}_j$. The variable $r$ indexes different operational equivalences. For every operational equivalence in $\mathcal{O}_P$ a corresponding ontological restriction must be satisfied by all noncontextual models: \begin{align*}
    \forall \lambda : \sum_i \alpha^r_i \mu_{P_i}(\lambda ) = \sum_j \beta^r_j \mu_{P_j}(\lambda ) 
\end{align*}
Likewise for measurements, a set $\mathcal{O}_\Mo$ will contain all operational equivalences of the form: \begin{align*}
    \sum_{i,j} \alpha^s_{i,j} \Mo_i(j) = \sum_{i',j'} \beta^s_{i',j'} \Mo_{i'}(j') 
\end{align*}for some set of convex weights $\{\alpha^s_{i,j}\}_{i,j}$ and $\{\beta^s_{i,j}\}_{i,j}$, where the variable $s$ indexes the set of operational equivalences.

Once the sets of operational equivalences $\mathcal{O}_P$ and $\mathcal{O}_M$ have been determined for a behavior $p(x|P_i , \Mo_j )$, the question if a noncontextual model exists for the behavior can be presented in a compact form.

\begin{definition}\label{def:noncontextual-model-exists}
A noncontextual ontological model exists for a behavior $p(x| P_i , \Mo_j)$ if there exists an ontic state space $\Lambda$, epistemic states $\{ \mu_{P_i}(\lambda )\}_i$ and response functions $\{ \xi_{\Mo_j}(k|\lambda ) \}_{j,k}$ such that: \begin{align*}
    \forall \lambda , j , k \,\,\,\,\, &  \xi_{\Mo_j}(k|\lambda ) \geq 0 \text{ (positivity of response functions)}\\
    \forall \lambda , j \,\,\,\,\, & \sum_k \xi_{\Mo_j}(k|\lambda ) = 1 \text{ (normalization of response functions)}\\
    \forall \lambda , s \,\,\,\,\, & \sum_{j,k}(\alpha^s_{j,k}-\beta^s_{j,k})\xi_{\Mo_j}(k|\lambda ) = 0 \text{ (operational equivalences $\mathcal{O}_\Mo$)}\\
    \forall \lambda , i \,\,\,\,\, & \mu_{P_i}(\lambda ) \geq 0 \text{ (positivity of epistemic states)}\\
    \forall i \,\,\,\,\, & \int_\Lambda \mu_{P_i}(\lambda ) d \lambda = 1 \text{ (normalization of epistemic states)}\\
    \forall \lambda , r \,\,\,\,\, & \sum_i ( \alpha^r_i - \beta^r_i ) \mu_{P_i}(\lambda ) = 0 \text{ (operational equivalences $\mathcal{O}_P$)}\\
    \forall i , j , k \,\,\,\,\, & \int_\Lambda \xi_{\Mo_j}(k| \lambda)\mu_{P_i}(\lambda )d\lambda = p(k| P_i, \Mo_j) \text{ (model reproduces data)}
\end{align*}
\end{definition}

A detailed description of methods that can be used to demonstrate a contextual advantage from any behavior is presented in \cite{ScSpWo2018}. We will use those methods to examine whether the communication tasks can be used to demonstrate a contextual advantage for quantum theory.

The first step is to characterize the so-called measurement assignment polytope. Suppose all measurements in a prepare-and-measure setup have $d$ outcomes and that there are $l$ distinct measurement procedures. Then the $ld$-dimensional vector $$\vec{\xi} = \begin{bmatrix} \xi_{\Mo_1}(1|\lambda^*) & \dots & \xi_{\Mo_1}(d|\lambda^*) & \xi_{\Mo_2}(1|\lambda^*) & \dots & \xi_{\Mo_l}(d|\lambda^*) \end{bmatrix},$$ defined for a specific ontic state $\lambda^*$, defines a noncontextual measurement assignment if it satisfies the first three conditions of Def. \ref{def:noncontextual-model-exists}. The set of all such assignments defines the measurement-assignment polytope which we must characterize by its vertices. The use of mathematical optimization software is encouraged to perform vertex enumeration\footnote{We used the function \texttt{fmel} from a free software called \texttt{PORTA}.}. The vertices are the extremal points of the convex measurement-assignment polytope.

Once vertex enumeration has been performed, a key observation is that any noncontextual model, no matter the size of the ontic state space $\Lambda$, can be reconstructed into a model defined by probability assignments to the finite set of extremal points of the measurement assignment polytope. That is, each preparation defines one epistemic state for each vertex in the measurement assignment polytope. Let $\kappa$ enumerate the vertices and let $\nu_P(\kappa)$ denote the epistemic states defined on the vertices.

\begin{definition}\label{def:noncontextual-model-exists-2}
A noncontextual ontological model exists for a behavior $p(x|P_i, \Mo_j  )$ if there exists epistemic states $\{ \nu_{P_i}(\kappa ) \}_{i,\kappa}$ such that: \begin{align*}
    \forall i , \kappa  \,\,\,\,\, & \nu_{P_i}(\kappa ) \geq 0 \text{ (positivity of epistemic states)}\\
    \forall i \,\,\,\,\, & \sum_\kappa \nu_{P_i}(\kappa ) = 1 \text{ (normalization of epistemic states)}\\
    \forall r , \kappa \,\,\,\,\, & \sum_i (\alpha^r_i - \beta^r_i ) \nu_{P_i}(\kappa ) = 0 \text{ (operational equivalences $\mathcal{O}_P$)}\\
    \forall i , j , k \,\,\,\,\, & \sum_\kappa \tilde{\xi}_{\Mo_j}(k|\kappa ) \nu_{P_i}(\kappa ) = p(k | P_i , \Mo_j), \text{ (model reproduces data)}
\end{align*} where $\tilde{\xi}_{\Mo_i}(j|\kappa )$ are the extremal response functions defined by the noncontextual measurement-assignment polytope.
\end{definition}

Another key observation from Def. \eqref{def:noncontextual-model-exists-2} is that each operational probability $p(k| P_i , \Mo_j)$ is a linear combination of a finite amount of variables. We could now proceed as in \cite{ScSpWo2018} and perform linear quantifier elimination in order to find all the noncontextual inequalities that define the whole noncontextual polytope of correlations. It's more convenient, however, to treat the equations of Def. \eqref{def:noncontextual-model-exists-2} as a linear program.

The last three equations of Def. \eqref{def:noncontextual-model-exists-2} can be encoded in a single matrix equation $\mathbf{M}\mathbf{x}=\mathbf{b}^*$, where $\mathbf{M}$ contains the operational equivalence parameters $\alpha^r_i - \beta^r_i$ and the extremal measurement assignments $\tilde{\xi}_{\Mo_j}(k|\kappa )$, $\mathbf{x}$ contains the epistemic states $\nu_{P_i}(\kappa )$ as a vector, and $\mathbf{b}^*$ contains the operational probabilities $p(k|P_i, \Mo_j)$ as well as zeros and ones corresponding to the normalization of epistemic states and operational equivalences $\mathcal{O}_P$. The positivity of epistemic states corresponds to the inequality $\mathbf{x} \geq 0$. Hence, the equations of Def. \eqref{def:noncontextual-model-exists-2} define a linear program (LP) whose primal feasibility is determined by the existence of $\mathbf{x}$ such that \begin{align}\label{eq:LP}\begin{split}
    \mathbf{M}\mathbf{x}&=\mathbf{b}^* \\
    \mathbf{x} &\geq 0
\end{split}\end{align}
Whenever such an $\mathbf{x}$ exists that Eq. \eqref{eq:LP} is fulfilled, we can be sure that a noncontextual ontological model exists for the given operational probabilities.

By Farkas' lemma \cite{Andersen2001}, it must be true that either the primal LP is feasible or a certificate of primal feasibility is negative. We can define the certificate of primal infeasibility as the solution to the Farkas dual of the primal LP: \begin{align}\label{eq:farkas-dual}
    \begin{split}
        \min_\mathbf{y} \mathbf{b^*}^{\top}\mathbf{y} \\
        \mathbf{M}^\top \mathbf{y} \geq 0.
    \end{split}
\end{align}
Whenever the primal LP is infeasible it's guaranteed that we can find a $\mathbf{y}$ such that $\mathbf{b^*}^{\top}\mathbf{y} < 0$. By finding the minimum value among these we will simultaneously find the noncontextuality inequality that is most violated by the operational probabilities.

\subsection{Grassmannian frames}

Let $V$ be an inner product space. A sequence $\fram$ is called a frame if there exists frame bounds $A$ and $B$ such that: \begin{align}
    A \no{v}^2 \leq \sum_{i=1}^n \abs{\innerp{v}{f_i}}^2 \leq B \no{v}^2
\end{align} for all $v \in V$. Some special cases of frames are listed below.

A frame is called: \begin{itemize}
    \item tight whenever the choice $A=B$ is possible.
    \item uniform whenever $\no{f_i} = 1$ for all $i$.
    \item equiangular whenever $\abs{\innerp{f_i}{f_j}} = c$ for all $i\neq j$ and $c\geq 0$. 
\end{itemize}
Additionally, whenever the number of frame elements $n$ coincides with the dimension $d$ of $V$, the frame is also a basis for $V$. An example of an equiangular uniform frame would be an orthonormal basis.

We are not particularly interested in frames that are also bases. There are many reasons to study overcomplete frames \cite{DuSc1952, Daubechies1992, BaCaHeLa2006}. An important concept is the correlation between frame elements.

\begin{definition}\label{def:maximal-frame-correlation}
Let $\fram$ be a uniform frame in $V$. Then \begin{align}
    \mathcal{M}\left( \fram \right) = \max_{j,k,j\neq k} \abs{\innerp{f_j}{f_k}}
\end{align}is the maximal frame correlation of $\fram$.
\end{definition}

The maximal frame correlation measures the maximal overlap between frame elements. A Grassmannian frame is simply a frame that minimizes the maximum overlap \cite{StHe2003, Leonhard2016}.

\begin{definition}\label{def:grassmannian-frame}
Let $\fram$ be a uniform frame in $V$. Then $\fram$ is called a Grassmannian frame if it is a solution to the problem \begin{align}
    \min \left\{  \mathcal{M}\left(  \fram \right) \right\}.
\end{align} 
\end{definition}
 At first glance it doesn't seem probable that there would exist analytical bounds on maximal frame correlations that only depend on the dimension of the inner product space and the number of frame elements. Yet such a result exists which we will make use of later on.
 
 \begin{proposition}\label{prop:frame-correlation-bounds}
 Let $\fram$ be a uniform frame for $\mathbb{C}^d$ or $\mathbb{R}^d$. Then \begin{align}\label{eq:frame-bound}
     \mathcal{M}\left( \fram \right) \geq \sqrt{\frac{n-d}{d(n-1)}}.
 \end{align}
 Moreover, equality in Eq. \eqref{eq:frame-bound} is achieved if and only if $\fram$ is a tight equiangular frame. Equality in Eq. \eqref{eq:frame-bound} requires $n \leq d^2$ for $\mathbb{C}^d$ and $n \leq \frac{d(d+1)}{2}$ for $\mathbb{R}^d$.
 \end{proposition}
 \noindent
A proof for Prop. \ref{prop:frame-correlation-bounds} can be found from \cite{Welch1974}.

Prop. \ref{prop:frame-correlation-bounds} concludes our brief introduction to frames. Additional information on frame theory can be found from various excellent sources, such as \cite{Christensen2008}.

\section{Bounds on communication success}\label{bounds}

We are now ready to analyse the tasks $T_{4,1}$ and $T_{4,2}$. Let us first consider the classical version with a single bit as the communication medium.

\subsection{Classical bounds}

We will follow a similar technique for the tasks $T_{4,1}$ and $T_{4,2}$ as was used for $T_{3,1}$. The task $T_{4,1}$ involves four possibilities for the strings $s$ and input $a$ while there are six possibilities for the input $b$. The whole strategy can be conveniently presented in table form.

\begin{table}
\centering
    \begin{tabular}{c|c}
        $a$ & $r(a)$ \\ \hline
        $1000$ & $i$ \\
        $0100$ & $j$ \\
        $0010$ & $k$ \\
        $0001$ & $l$
    \end{tabular}\qquad
    \begin{tabular}{c|c|c}
        $b$ & $g(b,0)$ & $g(b,1)$ \\ \hline
        $1100$ & $x_1$ & $y_1$ \\
        $1010$ & $x_2$ & $y_2$ \\
        $1001$ & $x_3$ & $y_3$ \\
        $0110$ & $x_4$ & $y_4$ \\
        $0101$ & $x_5$ & $y_5$ \\
        $0011$ & $x_6$ & $y_6$ 
    \end{tabular}
    
    \begin{tabular}{c|c|c|c|c}
        $s$ & $a$ & $b$ & $r(a)$ & $g(b, r(a))$ \\ \hline
        $1000$ & $0100$ & $0011$ & $j$ & $g(0011, j)$ \\
        $1000$ & $0010$ & $0101$ & $k$ & $g(0101, k)$ \\
        $1000$ & $0001$ & $0110$ & $l$ & $g(0110, l)$ \\
        $0100$ & $1000$ & $0011$ & $i$ & $g(0011, i)$ \\
        $0100$ & $0010$ & $1001$ & $k$ & $g(1001, k)$ \\
        $0100$ & $0001$ & $1010$ & $l$ & $g(1010, l)$ \\
        $0010$ & $1000$ & $0101$ & $i$ & $g(0101, i)$ \\
        $0010$ & $0100$ & $1001$ & $j$ & $g(1001, j)$ \\
        $0010$ & $0001$ & $1100$ & $l$ & $g(1100, l)$ \\
        $0001$ & $1000$ & $0110$ & $i$ & $g(0110, i)$ \\
        $0001$ & $0100$ & $1010$ & $j$ & $g(1010, j)$ \\
        $0001$ & $0010$ & $1100$ & $k$ & $g(1100, k)$ 
\end{tabular}\newline
\caption{General classical strategy for the task $T_{4,1}$.}
\label{tab:c41-general-strategy}
\end{table}
By looking at the rows of Tab. \ref{tab:c41-general-strategy} where the input $b$ is identical, we get the following constraints on the messages $r(a):$\begin{align*}
    i \neq j , \, i \neq k , \, i \neq l , \, j \neq k , \, j \neq l , \, l \neq k.
\end{align*}That means that, in order to never make a mistake, each message $r(a)$ needs to be distinct. However, the message  $r(a)$ can only take on the binary values of $0$ and $1$, and hence at least two of the inequalities must be violated. Each violation leads to one mistake, because each row in Tab. \ref{tab:c41-general-strategy} with identical $b$ has a different string $s$. Therefore any classical strategy formed by giving concrete values to Tab. \ref{tab:c41-general-strategy} must contain at least two mistakes. An optimal classical strategy for the task $T_{4,1}$ is presented in Tab. \ref{tab:c41-optimal-strategy}. The average classical success probability for the task $T_{4,1}$ is $5/6$.  

\begin{table}
\centering
    \begin{tabular}{c|c}
        $a$ & $r(a)$ \\ \hline
        $1000$ & $0$ \\
        $0100$ & $0$ \\
        $0010$ & $1$ \\
        $0001$ & $1$
    \end{tabular}\qquad
    \begin{tabular}{c|c|c}
        $b$ & $g(b,0)$ & $g(b,1)$ \\ \hline
        $1100$ & $-$ & $3$ \\
        $1010$ & $4$ & $2$ \\
        $1001$ & $3$ & $2$ \\
        $0110$ & $4$ & $1$ \\
        $0101$ & $3$ & $1$ \\
        $0011$ & $1$ & $-$ 
    \end{tabular}
    
    \begin{tabular}{c|c|c|c|c}
        $s$ & $a$ & $b$ & $r(a)$ & $g(b, r(a))$ \\ \hline
        $1000$ & $0100$ & $0011$ & $0$ & $1$ \\
        $1000$ & $0010$ & $0101$ & $1$ & $1$ \\
        $1000$ & $0001$ & $0110$ & $1$ & $1$ \\
        $0100$ & $1000$ & $0011$ & $0$ & $\red{1}$ \\
        $0100$ & $0010$ & $1001$ & $1$ & $2$ \\
        $0100$ & $0001$ & $1010$ & $1$ & $2$ \\
        $0010$ & $1000$ & $0101$ & $0$ & $3$ \\
        $0010$ & $0100$ & $1001$ & $0$ & $3$ \\
        $0010$ & $0001$ & $1100$ & $1$ & $3$ \\
        $0001$ & $1000$ & $0110$ & $0$ & $4$ \\
        $0001$ & $0100$ & $1010$ & $0$ & $4$ \\
        $0001$ & $0010$ & $1100$ & $1$ & $\red{3}$ 
\end{tabular}\newline
\caption{An optimal classical strategy for the task $T_{4,1}$. Red values in the last column indicate erroneous guesses by Bob. The guesses $g(1100, 0)$ and $g(0011,1)$ are not listed because they aren't needed in this specific implementation of the strategy.}
\label{tab:c41-optimal-strategy}
\end{table}

As the final classical example, let us introduce the task $T_{4,2}$. Now there are six possibilities for input $a$ while there are four possibilities for the string $s$ and input $b$.

\begin{table}
\centering
    \begin{tabular}{c|c}
        $a$ & $r(a)$ \\ \hline
        $1100$ & $i$ \\
        $1010$ & $j$ \\
        $1001$ & $k$ \\
        $0110$ & $l$ \\
        $0101$ & $m$ \\
        $0011$ & $n$ 
    \end{tabular}\qquad
    \begin{tabular}{c|c|c}
        $b$ & $g(b,0)$ & $g(b,1)$ \\ \hline
        $1000$ & $x_1$ & $y_1$ \\
        $0100$ & $x_2$ & $y_2$ \\
        $0010$ & $x_3$ & $y_3$ \\
        $0001$ & $x_4$ & $y_4$ 
    \end{tabular}
    
    \begin{tabular}{c|c|c|c|c}
        $s$ & $a$ & $b$ & $r(a)$ & $g(b, r(a))$ \\ \hline
        $1000$ & $0110$ & $0001$ & $l$ & $g(0001, l)$ \\
        $1000$ & $0101$ & $0010$ & $m$ & $g(0010, m)$ \\
        $1000$ & $0011$ & $0100$ & $n$ & $g(0100, n)$ \\
        $0100$ & $1010$ & $0001$ & $j$ & $g(0001, j)$ \\
        $0100$ & $1001$ & $0010$ & $k$ & $g(0010, k)$ \\
        $0100$ & $0011$ & $1000$ & $n$ & $g(1000, n)$ \\
        $0010$ & $1100$ & $0001$ & $i$ & $g(0001, i)$ \\
        $0010$ & $1001$ & $0100$ & $k$ & $g(0100, k)$ \\
        $0010$ & $0101$ & $1000$ & $m$ & $g(1000, m)$ \\
        $0001$ & $1100$ & $0010$ & $i$ & $g(0010, i)$ \\
        $0001$ & $1010$ & $0100$ & $j$ & $g(0100, j)$ \\
        $0001$ & $0110$ & $1000$ & $l$ & $g(1000, l)$ 
\end{tabular}\newline
\caption{General classical strategy for the task $T_{4,2}$.}
\label{tab:c42-general-strategy}
\end{table}
Again by looking at the rows of Tab. \ref{tab:c42-general-strategy} where the input $b$ is identical, we get the following constraints: \begin{align}
    i \neq j \neq l , \, i \neq k \neq m , \, j \neq k \neq n , \, l \neq m \neq n.
\end{align}
Any violation of the above inequalities instantly leads to at least one mistake in the strategy, as each inequality was derived from rows with different strings $s$. By checking all possible choices for the messages $r(a)$ we find that the minimum number of violations is four in any classical strategy. An example of such strategy is presented in Tab. \ref{tab:c42-optimal-strategy}. The average classical success probability for the task $T_{4,2}$ is therefore $2/3$ (4 mistakes out of twelve guesses).

\begin{table}
\centering
    \begin{tabular}{c|c}
        $a$ & $r(a)$ \\ \hline
        $1100$ & $0$ \\
        $1010$ & $1$ \\
        $1001$ & $1$ \\
        $0110$ & $1$ \\
        $0101$ & $1$ \\
        $0011$ & $0$ 
    \end{tabular}\qquad
    \begin{tabular}{c|c|c}
        $b$ & $g(b,0)$ & $g(b,1)$ \\ \hline
        $1000$ & $2$ & $3$ \\
        $0100$ & $1$ & $3$ \\
        $0010$ & $4$ & $1$ \\
        $0001$ & $3$ & $1$ 
    \end{tabular}
    
    \begin{tabular}{c|c|c|c|c}
        $s$ & $a$ & $b$ & $r(a)$ & $g(b, r(a))$ \\ \hline
        $1000$ & $0110$ & $0001$ & $1$ & $1$ \\
        $1000$ & $0101$ & $0010$ & $1$ & $1$ \\
        $1000$ & $0011$ & $0100$ & $0$ & $1$ \\
        $0100$ & $1010$ & $0001$ & $1$ & $\red{1}$ \\
        $0100$ & $1001$ & $0010$ & $1$ & $\red{1}$ \\
        $0100$ & $0011$ & $1000$ & $0$ & $2$ \\
        $0010$ & $1100$ & $0001$ & $0$ & $3$ \\
        $0010$ & $1001$ & $0100$ & $1$ & $3$ \\
        $0010$ & $0101$ & $1000$ & $1$ & $3$ \\
        $0001$ & $1100$ & $0010$ & $0$ & $4$ \\
        $0001$ & $1010$ & $0100$ & $1$ & $\red{3}$ \\
        $0001$ & $0110$ & $1000$ & $1$ & $\red{3}$ 
\end{tabular}\newline
\caption{An optimal classical strategy for the task $T_{4,2}$. The red values indicate erroneous guesses by Bob.}
\label{tab:c42-optimal-strategy}
\end{table}

\subsection{Contextual bounds}

By introducing operational equivalences to Alice's preparations\footnote{We could also consider operational equivalences for Bob's measurements. However, it's sufficient to only consider operational equivalences for preparations in order to observe a quantum advantage.} we can produce noncontextuality inequalities that bound the success chance of noncontextual models on communication tasks of any type with $n \geq 4$. In this way we can compare noncontextual models to the simplest case of one bit as the communication medium, while simultaneously observing a quantum advantage using SDP methods.

In the partial ignorance communication task of type $T_{4,1}$ Charlie sends one of four strings to Alice: $1000$, $0100$, $0010$ or $0001$. Alice prepares a state from the set $\{ \rho_i \}_{i=1}^4$, where the subscript $i$ corresponds with the index of the 1 in Alice's input string. Bob will receive one of the following input strings from Charlie: $1100$, $1010$, $1001$, $0110$, $0101$ or $0011$. For each input string Bob will have a binary POVM $\Mo_{ij}$, where the subscript $ij$ will indicate the indexes of the 0's in Bob's input. When Bob's measurements are labeled in this way all Bob's POVMs will contain the index of Alice's preparation in the label. In other words, the success metric for the task $T_{4,1}$ will be a combination of terms $\pm \tr{\rho_i\Mo_{jk}(1)}$ where either $i=j$ or $i=k$. The term will be positive if $i=k$ and negative if $i=j$. By maximizing this success metric Bob's output will coincide with the correct answer as often as possible on average.

\begin{example}{\textit{(Contextual $T_{4,1}$)}}\label{ex:contextual-c41}
Let us consider the task $T_{4,1}$ with a non-trivial operational constraint between Alice's preparations. As there are four possible states that Alice can prepare, the natural choice for the operational equivalence is $\rho_1 + \rho_2 \simeq \rho_3 + \rho_4$. We wish to optimize the following problem:
\begin{subequations}
\begin{align}\label{eq:c-412formulation}
\begin{split}
\max  \quad & -\tr{\rho_1 \Mo_{12}(1)} - \tr{\rho_1 \Mo_{13}(1)} - \tr{\rho_1 \Mo_{14}(1)} \\
& + \tr{\rho_2 \Mo_{12}(1)} - \tr{\rho_2 \Mo_{23}(1)} - \tr{\rho_2 \Mo_{24}(1)} \\
& + \tr{\rho_3 \Mo_{13}(1)} + \tr{\rho_3 \Mo_{23}(1)} - \tr{\rho_3 \Mo_{34}(1)} \\
& + \tr{\rho_4 \Mo_{14}(1)} + \tr{\rho_4 \Mo_{24}(1)} + \tr{\rho_4 \Mo_{34}(1)} 
\end{split}\\
\text{subject to} \quad & \rho_i \geq 0, \quad  i=1,2,3,4 \\
& \tr{\rho_i}=1, \quad i=1,2,3,4 \\
& \rho_1 + \rho_2 = \rho_3 + \rho_4 \label{eq:c-412formulation-op}\\
& \Mo_j(1),\Mo_j(2) \geq 0, \quad j=12,13,14,23,24,34 \\
& \Mo_j(1) + \Mo_j(2) = \id , \quad j=12,13,14,23,24,34.
\end{align}
\end{subequations}
Note that the objective function is written in terms of the first effect of each POVM.

The unitary hierarchy $\mathcal{U}_1$ converges on $4.828427123$, while the see-saw algorithm converges on $4.828427104$ for qubits. Hence the outer bound on the first level of the unitary SDP hierarchy matches the inner bound for qubits by 7 decimals. Note that the outer bound doesn't depend on dimension. Instead it is valid for all quantum systems that conform to the operational equivalence between preparations. The average success probability, as given by the implementation found by the see-saw algorithm, is $0.9023689$, much higher than the average classical success probability.

Performing vertex enumeration we find that the measurement assignment polytope has 64 vertices. Hence there are $256$ variables $\nu_{P_i}(\kappa )$. We then from the matrix $\mathbf{M}$ and query a linear program for the optimal solution to Eq. \eqref{eq:c-412formulation} within noncontextual ontological models. We find that the answer is 4 within numerical accuracy. By forming the Farkas' dual we find that the noncontextual inequality that is most violated by the optimal quantum implementation is: \begin{align*}
     &+ 0.000764p_{1, 12} + 0.000764p_{2, 12} - 0.000764p_{3, 12} - 0.000764p_{4, 12} \\
     &+ 0.250121p_{1, 13} - 0.249879p_{2, 13} - 0.250121p_{3, 13} + 0.249879p_{4, 13} \\
     &+ 0.249258p_{1, 14} - 0.249258p_{2, 14} + 0.249258p_{3, 14} - 0.249258p_{4, 14} \\
     &- 0.250742p_{1, 23} + 0.250742p_{2, 23} - 0.250742p_{3, 23} + 0.250742p_{4, 23} \\
     &- 0.250121p_{1, 24} + 0.249879p_{2, 24} + 0.250121p_{3, 24} - 0.249879p_{4, 24} \\
     &- 0.000764p_{1, 34} - 0.000764p_{2, 34} + 0.000764p_{3, 34} + 0.000764p_{4, 34} \leq 1,
\end{align*}
where $p_{i,jk}$ is short-hand for $p(1| P_i , \Mo_{jk} )$. Our quantum implementation achieves a value of $1.414213561$ for this inequality, or $\sqrt{2}$ within numerical accuracy, a violation of $\sqrt{2}-1$ for the noncontextual bound.

As a final observation we can directly compare the noncontextual bound on the success metric of Eq. \eqref{eq:c-412formulation} to the optimal bit implementation as follows. The terms in the success metric in Eq. \eqref{eq:c-412formulation} are written in the same order as the rows in Tab. \ref{tab:c41-optimal-strategy}. A correct value in the bit implementation corresponds to a $1$ or $0$ depending on the sign of the term in the success metric. An incorrect value in the bit implementation leads to a degradation of the success metric by 1. As it happens the bit implementation also obtains a value of 4 for the success metric.
\end{example}

In the partial ignorance communication task of type $T_{4,2}$ Charlie still has 4 input strings $s$: $1000$, $0100$, $0010$ and $0001$. The inputs of Alice and Bob are reversed in a sense. Alice will receive one of the following input strings: $1100$, $1010$, $1001$, $0110$, $0101$ and $0011$. Bob, on the other hand, will receive one of the following input strings: $1000$, $0100$, $0010$ or $0001$. Hence Alice will need to prepare states with 6 distinct labels although some of these states may be identical, and Bob will need to perform four ternary measurements. Alice will label her states as $\rho_{ij}$, where the subscript $ij$ will indicate the indices of the 1's in her input string. Bob will label his POVMs as $\Mo_k$, where the $k$ will correspond to the index of the 1 in Bob's input. Each of Bob's POVM $\Mo_k$ will have the set $\{ 1,2,3,4 \} \setminus k$ as the outcome set. Thus upon receiving measurement outcome $l$ Bob will guess that the index of the 1 in $s$ was $l$. The success metric that Alice and Bob try to maximize is $$\sum_{\substack{i,j,k,l=1 \\ i \neq j \neq k \neq l}}^4 \tr{\rho_{ij}\Mo_k(l)}.$$ We will, however, relabel the effects with the following rule: $$\Mo_k(l) = \begin{cases} \Mo_k(l), &  \mbox{ if } l < k \\ \Mo_k(l-1), & \mbox{ if } l > k.\end{cases}$$ With this relabeling all Bob's measurements will have the outcome set $\{ 1,2,3\}$. This makes the SDPs somewhat easier to program, but as a drawback Bob's outcome will not directly correspond with his guess, but has to be interpreted according to the relabeling.

\begin{example}{\textit{(Contextual $T_{4,2}$)}}\label{ex:contextual-c42}
The success metric for the task $T_{4,2}$, written in terms of the first two effects, is written as:
\begin{subequations}
\begin{align}\label{eq:c42-formulation}
\begin{split}
\max  \quad & \tr{\rho_{23} \Mo_4(1)} + \tr{\rho_{24} \Mo_3(1)} + \tr{\rho_{34} \Mo_2(1)} \\
& + \tr{\rho_{13} \Mo_4(2)} + \tr{\rho_{14} \Mo_3(2)} + \tr{\rho_{34} \Mo_1(1)} \\
& + 1 - \tr{\rho_{12} (\Mo_4(1) + \Mo_4(2))} + \tr{\rho_{14} \Mo_2(2)} + \tr{\rho_{24} \Mo_1(2)} \\
& + 3 - \tr{\rho_{12} ( \Mo_3(1) + \Mo_3(2) )} - \tr{\rho_{13} ( \Mo_2(1) + \Mo_2(2) )} \\
& - \tr{\rho_{23} ( \Mo_1(1) + \Mo_1(2) ) } 
\end{split}\\
\text{subject to} \quad & \rho_{ij} \geq 0, \quad  ij=12,13,14,23,24,34 \\
& \tr{\rho_{ij}}=1, \quad ij=12,13,14,23,24,34 \\
& \rho_{12} + \rho_{13} = \rho_{14} + \rho_{23} = \rho_{24} + \rho_{34} \label{eq:c42-formulation-op}\\
& \Mo_k(1),\Mo_k(2), \Mo_k(3) \geq 0, \quad k=1,2,3,4 \\
& \Mo_k(1) + \Mo_k(2) + \Mo_k(3) = \id , \quad k=1,2,3,4.
\end{align}
\end{subequations}

The order of terms in the success metric of Eq. \eqref{eq:c42-formulation} is the same is the rows in Tab. \ref{tab:c42-optimal-strategy}. A nontrivial operational equivalence between preparations is defined in Eq. \eqref{eq:c42-formulation-op}

Querying a linear program for the maximal value to Eq. \eqref{eq:c42-formulation} we get a maximal value of 8 within numerical accuracy for all noncontextual models. Interestingly this again coincides with the optimal bit implementation of Tab. \ref{tab:c42-optimal-strategy}. The see-saw method also converges on exactly 8 for both qubits and qutrits. This leads us to suspect that no contextual advantage is possible in this task. The unitary hierarchy $\mathcal{U}_1$ seems to break down for this task, so we will need to prove the outer bound with analytical methods.
\end{example}

\begin{proposition}\label{prop:c42-prop}
The outer bound for the success metric of Eq. \eqref{eq:c42-formulation} equals exactly 8.
\end{proposition}
\begin{proof}
To prove the outer bound we will write the outcome distributions of the first POVM as a communication matrix: \begin{align*}
    A_1=\begin{tabular}{c|c c c}
         & $\Mo_1(1)$ & $\Mo_1(2)$ & $\Mo_1(3)$ \\ \hline
         $\rho_{12}$ & $\tr{\rho_{12}\Mo_1(1)}$ & $\tr{\rho_{12}\Mo_1(2)}$ & $\tr{\rho_{12}\Mo_1(3)}$ \\ 
         $\rho_{13}$ & $\tr{\rho_{13}\Mo_1(1)}$ & $\tr{\rho_{13}\Mo_1(2)}$ & $\tr{\rho_{13}\Mo_1(3)}$ \\ 
         $\rho_{14}$ & $\tr{\rho_{14}\Mo_1(1)}$ & $\tr{\rho_{14}\Mo_1(2)}$ & $\tr{\rho_{14}\Mo_1(3)}$ \\ 
         $\rho_{23}$ & $\tr{\rho_{23}\Mo_1(1)}$ & $\tr{\rho_{23}\Mo_1(2)}$ & \blue{$\tr{\rho_{23}\Mo_1(3)}$} \\ 
         $\rho_{24}$ & $\tr{\rho_{24}\Mo_1(1)}$ & \blue{$\tr{\rho_{24}\Mo_1(2)}$} & $\tr{\rho_{24}\Mo_1(3)}$ \\ 
         $\rho_{34}$ & \blue{$\tr{\rho_{34}\Mo_1(1)}$} & $\tr{\rho_{34}\Mo_1(2)}$ & $\tr{\rho_{34}\Mo_1(3)}$ \\ 
    \end{tabular}
\end{align*}
The matrix $A_1$ is row-stochastic, each row sums up to 1. The terms of $A_1$ that show up in the success metric Eq. \eqref{eq:c42-formulation} are highlighted in blue. Those three terms are the ones that we need to maximize in this communication matrix. From the operational equivalence $\rho_{14}+\rho_{23}\simeq\rho_{24}+\rho_{34}$ it follows that \begin{align*}
    2 &= \tr{\rho_{24}(\Mo_1(1)+\Mo_1(2)+\Mo_1(3))}+\tr{\rho_{34}(\Mo_1(1)+\Mo_1(2)+\Mo_1(3))}\\
    &= \tr{\rho_{24}(\Mo_1(1)+\Mo_1(2))}+\tr{\rho_{34}(\Mo_1(1)+\Mo_1(2))}+\tr{(\rho_{14}+\rho_{23})\Mo_1(3)}\\
    &\geq \tr{\rho_{34}\Mo_1(1)}+\tr{\rho_{24}\Mo_1(2)}+\tr{\rho_{23}\Mo_1(3)}
\end{align*}
Hence in any operational theory that satisfies the operational equivalence $\rho_{14}+\rho_{23}\simeq\rho_{24}+\rho_{34}$ the inequality $\tr{\rho_{34}\Mo_1(1)}+\tr{\rho_{24}\Mo_1(2)}+\tr{\rho_{23}\Mo_1(3)} \leq 2$ holds. It turns out that a similar argument holds for the other POVMs as well so that the maximum value for Eq. \eqref{eq:c42-formulation} is 8 for all operational theories that that satisfies the operational equivalence \eqref{eq:c42-formulation-op}.
\end{proof}

The previous proposition concludes our contextual examples. In the next section we will study the tasks $T_{4,1}$ and $T_{4,2}$ without the operational equivalences.

\subsection{General bounds}

We can produce inner bounds on the tasks $T_{4,1}$ and $T_{4,2}$ without the operational equivalences \eqref{eq:c-412formulation-op} and \eqref{eq:c42-formulation-op} by using the see-saw algorithm. Tab. \ref{tab:see-saw-results} collects these results.

\begin{table}[ht]
    \centering
\begin{tabular}{c|c|c|c}
    $d$ & 2 & 3 & 4 \\ \hline
    $T_{4,1}$ & 10.89897946 & 11.65685425 & 12.0 \\
    $T_{4,2}$ & 8.0 & 12.0 & 12.0
\end{tabular}
    \caption{Inner bounds on the tasks $T_{4,1}$ and $T_{4,2}$ as produced by the see-saw algorithm.}
    \label{tab:see-saw-results}
\end{table}
The bounds on $T_{4,1}$ in Tab. \ref{tab:see-saw-results} are presented in terms of both effects, i.e., the negative terms in Eq. \eqref{eq:c-412formulation} are replaced by the corresponding positive term.

We can see from Tab. \ref{tab:see-saw-results} that  there is a major discrepancy between the tasks $T_{4,1}$ and $T_{4,2}$. The inner bounds on task $T_{4,1}$ gradually increase from the minimum value obtained by qubits to the maximum value of $12$ obtained by four-dimensional quantum states. The value obtained by qubits is slightly greater than the contextual limit for the task.

For the task $T_{4,2}$ the inner bound is 8 for qubits within numerical accuracy, the same as the outer bound found in Prop. \ref{prop:c42-prop}. For qutrits the inner bound already reaches the maximum value of 12. We will now try to make sense of these numbers.

\begin{definition}
Let $A$ be a communication matrix, i.e., $A_{ij} = \tr{\rho_i\Mo(j)}$ for some finite set of quantum states $\{ \rho_i \}_i$ and a POVM $\Mo$ with a finite outcome set. The function $\lambda_{max}$ is defined as $\lambda_{max}(A) := \sum_j \max_i(A_{ij})$.
\end{definition}

The function $\lambda_{max}$ is an ultraweak monotone on the set of communication matrices. Ultraweak monotones were extensively studied in \cite{HeKeLe2020}. However, for our present investigation a more important result is the following. \begin{proposition}\label{prop:lambda-max}
Let $A$ be a communication matrix. Then \begin{align*}
    \mathrm{rank}_{psd}(A) \geq \lambda_{max}(A),
\end{align*}where $\mathrm{rank}_{psd}$ is the positive semidefinite rank of a matrix.
\end{proposition}
The above proposition was proved in \cite{LeWeWo2017}. For communication matrices it holds that a necessary and sufficient condition for a communication matrix to have a $d$-dimensional quantum implementation is that the positive semidefinite rank of the matrix is not greater than $d$ \cite{HeKeLe2020}. With this knowledge we can prove the outer bound on task $T_{4,2}$.

\begin{proposition}\label{prop:c42-qubits}
The outer bound on task $T_{4,2}$ is 8 for qubits.
\end{proposition}
\begin{proof}
The task $T_{4,2}$ consists of four communication matrices. The communication matrix of the first POVM is the following: \begin{align*}
    A_1=\begin{tabular}{c|c c c}
         & $\Mo_1(1)$ & $\Mo_1(2)$ & $\Mo_1(3)$ \\ \hline
         $\rho_{12}$ & $\tr{\rho_{12}\Mo_1(1)}$ & $\tr{\rho_{12}\Mo_1(2)}$ & $\tr{\rho_{12}\Mo_1(3)}$ \\ 
         $\rho_{13}$ & $\tr{\rho_{13}\Mo_1(1)}$ & $\tr{\rho_{13}\Mo_1(2)}$ & $\tr{\rho_{13}\Mo_1(3)}$ \\ 
         $\rho_{14}$ & $\tr{\rho_{14}\Mo_1(1)}$ & $\tr{\rho_{14}\Mo_1(2)}$ & $\tr{\rho_{14}\Mo_1(3)}$ \\ 
         $\rho_{23}$ & $\tr{\rho_{23}\Mo_1(1)}$ & $\tr{\rho_{23}\Mo_1(2)}$ & \blue{$\tr{\rho_{23}\Mo_1(3)}$} \\ 
         $\rho_{24}$ & $\tr{\rho_{24}\Mo_1(1)}$ & \blue{$\tr{\rho_{24}\Mo_1(2)}$} & $\tr{\rho_{24}\Mo_1(3)}$ \\ 
         $\rho_{34}$ & \blue{$\tr{\rho_{34}\Mo_1(1)}$} & $\tr{\rho_{34}\Mo_1(2)}$ & $\tr{\rho_{34}\Mo_1(3)}$ \\ 
    \end{tabular}
\end{align*}
The blue color highlights those elements of $A_1$ that should be maximized. If $A_1$ is implementable with qubits, it follows that $\mathrm{rank}_{psd}(A_1) \leq 2$. But then by Prop. \ref{prop:lambda-max}, $\lambda_{max}(A_1)\leq 2$. This means that the sum of the elements highlighted in blue cannot be greater than 2. A similar argument holds for the other POVMs so the total outer bound on task $T_{4,2}$ for qubits is exactly 8.
\end{proof}

We could use Prop. \ref{prop:c42-qubits} to produce the outer bound 12 for qutrits. However, 12 is already the maximum value that the success metric can obtain, so there is no need to prove the outer bound for qutrits.

Let us now return to the task $T_{4,1}$. In this task Bob is using dichotomic measurements, that is, Bob is always trying to distinguish between two possible input states from Alice. The best strategy therefore tries to minimize the maximal overlap between Alice's states, so that they can be distinguished pair-wise as well as possible. This ensures there won't be any weakness in Alice and Bob's strategy that Charlie could exploit. We need just one more definition before we can prove that the bounds in Tab. \ref{tab:see-saw-results} are tight for $T_{4,1}$.

\begin{definition}\label{def:ambiguous-distinguishability}
The ambiguous distinguishability of two pure states $\kb{\varphi_1}{\varphi_1}$ and $\kb{\varphi_2}{\varphi_2}$ is defined as \begin{align*}
    \psuc = \frac 12 \left( 1 + \sqrt{1-\abs{\ip{\varphi_1}{\varphi_2}}^2} \right)
\end{align*}
\end{definition}

\begin{proposition}
The inner bounds presented in Tab. \ref{tab:see-saw-results} are tight for the task $T_{4,1}$.
\end{proposition}
\begin{proof}
Notice that the success metric \eqref{eq:c-412formulation} consists of 6 terms of the form $\tr{\rho_j\Mo_{ij}(1)}-\tr{\rho_i\Mo_{ij}(1)}$. We can interpret this term as $\psuc - \pfail =2\psuc -1$ for the POVM $\Mo_{ij}$.  Let $\{ \varphi_i \}_{i=1}^4$ be a Grassmannian frame for their span in $\mathbb{C}^d$. Then $\psuc \leq \frac 12 \left( 1 + \sqrt{1-\frac{4-d}{3d}} \right)$ by Prop. \ref{prop:frame-correlation-bounds}. Equality can be achieved if and only if $\{ \varphi_i \}_{i=1}^4$ is equiangular and that is the best possible success chance for Bob while minimizing the maximum overlap between Alice's states. For qubits we get $\psuc \leq \frac 12 \left(1+\sqrt{\frac 23}\right)$. Omitting the $-1$ terms we get an outer bound of $12\psuc =6 \left(1+\sqrt{\frac 23}\right) \approx 10.89897949$, matching the value in Tab. \ref{tab:see-saw-results} by 6 decimals. For qutrits the corresponding values are $\psuc \leq \frac 12 \left( 1 + \frac{2\sqrt 2}{3} \right)$ and outer bound $12\psuc =6 \left(1 + \frac{2\sqrt 2}{3}\right) \approx 11.65685425$, again matching the value in Tab. \ref{tab:see-saw-results} within numerical accuracy. For $d=4$ the maximal value of $12$ is achieved as there are 4 distinguishable states.
\end{proof}

\section{Final thoughts}\label{final-thoughts}

In this article we have continued the work done in \cite{HeKe2019} by introducing an input for Bob in the setting of communication of partial ignorance. We called these new types of communication tasks simply partial ignorance communication tasks of type $T_{n,m}$, where $n$ is an integer representing the length of Charlie's string $s$ and $m$ represents loosely the amount of information on $s$ that Charlie reveals to Alice. Whatever information Charlie doesn't reveal to Alice he will instead reveal to Bob, so that the partial ignorance communication task of type $T_{n,m}$ is always informationally complete, that is, Alice and Bob's combined knowledge is enough to solve the communication task perfectly.

We began by analysing the simplest communication task of type $T_{3,1}$ for bits and qubits. This proved to be a relatively straight-forward task. The tasks $T_{4,1}$ and $T_{4,2}$ were much more complicated. We used various methods to establish inner and outer bounds on success metrics. These methods included SDPs, ultraweak monotones on communication matrices and frame theory for quantum states. We found out that the bit was as good as a communication medium as any noncontextual ontological model for both tasks $T_{4,1}$ and $T_{4,2}$. For $T_{4,1}$ we observed a contextual advantage for qubits and proved tight general bounds on the success metric for qubits and qutrits. For the task $T_{4,2}$ we proved that a contextual advantage was impossible for the operational equivalence we chose and that the qubit didn't perform any better than the bit.

The complexity involved in solving the inner and outer bounds for communication tasks of type $T_{n,m}$ grows rapidly with $n$. Actually already for $T_{4,2}$ we could have chosen another set of operational equivalences for preparations and possibly included operational equivalences for measurements as well. However, we believe that the methods presented in this article serve as a good starting point when solving communication tasks similar to those studied in this work.

There were several questions that ultimately couldn't be included in the scope of this article. The first question is what happens if we abandon the requirement of informational completeness so that Alice and Bob's combined knowledge doesn't determine the correct answer completely. This would certainly increase the complexity involved in solving for the optimal strategies, but we believe the current methods would suffice to analyse those cases as well. The second question involves shared randomness. It's known that shared randomness is a powerful resource for communication \cite{FrWe2015}. Alice and Bob's access to shared randomness would most certainly increase the effectiveness of their best strategies, as this would allow them to mix strategies in a way that would at the very least increase the worst case success probabilities for the bit. However, the effect of shared randomness is already tedious to analyse at the level of communication matrices. We leave this case for future research.

As we discussed earlier, the bit turned out to be as effective as a communication medium as any noncontextual ontological model in the tasks $T_{4,1}$ and $T_{4,2}$. It's not entirely clear to us why this is the case. The bit doesn't respect any operational equivalences as there are only two possible states and a measurement simply consists of receiving the bit that was sent. In this sense the bit makes up a poor comparison for noncontextual models. Nevertheless the average success probabilities are the same. However, it must be kept in mind that the worst case success probability is always zero for the bit. The worst case success probability is greater than zero for qubits in the task $T_{4,2}$, but it's not clear what the actual worst case probability is as the see-saw method is only capable of optimizing the average success chance (linear functions of states and effects). Frame theory doesn't help either because equality in Eq. \eqref{eq:frame-bound} cannot be achieved. We leave the worst case success chance for qubits in $T_{4,2}$ as an open problem.

As a final question we could consider possible connections of the communication tasks studied in this work to other similar communication tasks. We believe that some partial ignorance communication tasks could be mapped to a corresponding random access code. It might even be possible to map all RACs to some subset of partial ignorance communication tasks of particular types. However, this connection is not clear to us at present time and requires further thought.

\section*{Acknowledgement}
O.K. personally thanks Teiko Heinosaari for his support during the research for this article and his insightful comments. O.K. acknowledges financial support from the Turku University Foundation during the research of this article.

\end{document}